\documentclass[conference]{IEEEtran}
\IEEEoverridecommandlockouts

\usepackage{cite}
\usepackage{amsmath,amssymb,amsfonts,amsthm}
\usepackage{graphicx}
\usepackage{textcomp}
\usepackage{xcolor}
\usepackage{algorithm}
\usepackage{algpseudocode}
\usepackage{bm}
\usepackage{microtype}

\allowdisplaybreaks

\def\BibTeX{{\rm B\kern-.05em{\sc i\kern-.025em b}\kern-.08em
    T\kern-.1667em\lower.7ex\hbox{E}\kern-.125emX}}

\newcommand{\Tr}{\operatorname{Tr}}

\newcommand{\cX}{\mathcal{X}}

\newcommand{\KL}{D_{\mathrm{KL}}}
\newcommand{\Dens}{\mathcal{S}}
\newcommand{\ip}[2]{\langle #1, #2 \rangle}
\newcommand{\D}{\mathrm{D}}

\theoremstyle{plain}
\newtheorem{theorem}{Theorem}
\newtheorem{lemma}{Lemma}
\newtheorem{corollary}{Corollary}
\newtheorem{assumption}{Assumption}
\theoremstyle{definition}
\newtheorem{definition}{Definition}
\newtheorem{remark}{Remark}

\begin{document}

\title{A Mirror-Descent Algorithm for Computing the Petz--R\'enyi Capacity of Classical-Quantum Channels}

\author{\IEEEauthorblockN{Yu-Hong Lai and Hao-Chung Cheng}
\IEEEauthorblockA{\textit{Department of Electrical Engineering},
\textit{National Taiwan University},
Taipei 10617, Taiwan \\
\textit{Physics/Math Division, National Center for Theoretical Sciences}, Taipei 10617, Taiwan \\
\textit{Hon Hai (Foxconn) Quantum Computing Center}, New Taipei City 236, Taiwan
}
}

\maketitle

\begin{abstract}
We study the computation of the $\alpha$-R\'enyi capacity of a classical-quantum (c-q) channel for $\alpha\in(0,1)$. 
We propose an exponentiated-gradient (mirror descent) iteration that generalizes the Blahut--Arimoto algorithm. Our analysis establishes relative smoothness with respect to the entropy geometry, guaranteeing a global sublinear convergence of the objective values. 
Furthermore, under a natural tangent-space nondegeneracy condition (and a mild spectral lower bound in one regime), we prove local linear (geometric) convergence in Kullback--Leibler divergence on a truncated probability simplex, with an explicit contraction factor once the local curvature constants are bounded.
\end{abstract}


\section{Introduction}
The Shannon capacity of a channel characterizes its fundamental capability of information transmission.
This quantity is given by maximizing the mutual information between the channel input and output, over all possible input distributions.
Although Shannon's capacity is itself an optimization problem, Blahut--Arimoto (BA) algorithm was devised to compute it in an iterative manner \cite{blahut1972computation, arimoto1972algorithm}.
Variants of the BA algorithms were also proposed and analyzed; see, e.g., \cite{kamatsuka2024new} and references therein.
In particular, for classical order-$\alpha$ capacities (e.g.\ Sibson/Arimoto capacities), BA-type alternating optimization schemes admit global convergence guarantees; \cite{kamatsuka2024new} shows equivalence among several such iterations under appropriate conditions on their initial distributions and proves global convergence to the optimum.
It was also generalized, in a non-trivial way, to compute the capacity of classical-quantum (c-q) channels \cite{li2020cqba, Ramakrishnan_2020, He_2024}.

In this work, we study computation of a one-parameter family, called Petz--R\'enyi capacity \cite{arimoto1977orderalpha, nakiboglu_2019} for c-q channels.
Such a quantity recently became important because of its role in characterizing the error exponent (or called reliability function) for c-q channels \cite{BH98, Hol00, Dal13, Ren25, LY25, preparation}.
For an input distribution $p$ on a finite alphabet $\mathcal{X}$ and a c-q channel $W:\cX \ni x \mapsto W_x \in \Dens(\mathbb{C}^{d}) \triangleq \left\{\rho: \rho \in \mathbb{C}^{d\times d}, \rho \geq 0, \Tr\left[\rho\right]=1\right\}$,
we define the (order-$\alpha$) Petz--R\'enyi information for $\alpha \in (0,1)$ as
\begin{align}
I_\alpha(p;W)
&\triangleq \inf_{\sigma \in \Dens(\mathbb{C}^{d})} \frac{1}{\alpha-1} \log \sum_{x\in\mathcal{X}} p_x \, e^{(\alpha-1)D_{\alpha}(W_x \Vert \sigma)}
\\
&=\frac{\alpha}{\alpha-1}\,
\!\log \Tr\!\left[\left(\sum_{x\in\cX} p_x\, W_x^\alpha\right)^{\!\frac{1}{\alpha}}\right],
\label{eq:intro_Ialpha}
\end{align}
where $D_{\alpha}(\rho\Vert\sigma) \triangleq \frac{1}{\alpha-1} \log \Tr\left[\rho^{\alpha}\sigma^{1-\alpha}\right]$, $\alpha \in (0,1)$, is the Petz--R\'enyi divergence \cite{petz1986quasi}, and \eqref{eq:intro_Ialpha} follows from quantum Sibson's identity \cite{KW09, SW12, HT14, CGH18}.
The Petz--R\'enyi capacity is then defined by optimizing $I_\alpha(p;W)$ over the probability simplex,
\begin{align}
C_\alpha(W)\triangleq \max_{p\in\Delta(\cX)} I_\alpha(p;W), \quad\alpha \in (0,1).
\end{align}

\noindent\textbf{Contributions.}
We study the computation of the Petz--R\'enyi capacity of  classical--quantum channels and cast it as minimizing a convex trace-power objective over the probability simplex. Using the negative-entropy geometry, we derive an entropic mirror-descent method with a closed-form exponentiated-gradient update, resulting in a Blahut--Arimoto--type iteration with a simple constant-stepsize implementation. Our analysis develops sharp spectral divided-difference bounds for the Hessian of the trace-power objective, which yields relative smoothness with respect to entropy and a global non-asymptotic sublinear decrease in objective values. 
Finally, on a truncated simplex and under a natural tangent-space nondegeneracy condition, we establish a matching local curvature lower bound and obtain local linear (geometric) convergence in the Kullback--Leibler (KL) divergence.

The Petz--R\'enyi capacity admits an equivalent form of maximizing the so-called Petz--Augustin information \cite{nakiboglu_augustin_2018, MO17, CGH18}:
\begin{align}
C_{\alpha}(W)
&= \max_{p\in\Delta(\mathcal{X})} I_{\alpha}^{\texttt{Aug}}(p;W),
\\
I_{\alpha}^{\texttt{Aug}}(p;W)
&\triangleq \inf_{\sigma \in \mathcal{S}(\mathbb{C}^d)} \sum\nolimits_{x\in\mathcal{X}} p_x \cdot D_{\alpha}(W_x\Vert \sigma). \label{eq:Augustin}
\end{align}
A recent work \cite{chu2025petzaugustin} develops a fixed-point method for computing the minimizer in \eqref{eq:Augustin}  and applies it to Petz--R\'enyi capacity for $\alpha\in(1/2,1)$ with a linear rate.
Our approach of maximizing the Petz--R\'enyi information in \eqref{eq:intro_Ialpha} avoids such an inner minimization.

\medskip
This paper is organized as follows.
We introduce our problem setup in Section~\ref{sec:setup} and the proposed algorithm in Section~\ref{sec:algorithm}.
In Section~\ref{sec:analysis_hessian}, we develop exact Hessian expressions and spectral kernel bounds that serve as the main technical tool.
Section~\ref{sec:smoothness} uses these bounds to prove entropy-relative smoothness and obtain global sublinear convergence rates.
Section~\ref{sec:linear} establishes relative strong convexity on a truncated simplex under a nondegeneracy condition, yielding local linear convergence in KL divergence.
We report numerical results in Section~\ref{sec:num}.

\smallskip
\noindent\textbf{Note}.
An independent and concurrent work \cite{chu2026capacity} also studies computation of the Petz--R\'enyi capacity for $\alpha \in [1/2,1]$. We note that their algorithm and analysis are different from ours.

\section{Setup and Convexity}
\label{sec:setup}

Fix $\alpha\in(0,1)$ and set $\beta\triangleq 1/\alpha>1$. For $p\in\Delta(\cX)$, define the operators $A_x \triangleq W_x^{\alpha}$ and
\begin{equation}
\label{eq:def_M_S}
M(p)\triangleq \sum_{x\in\cX} p_x A_x,\qquad
S(p)\triangleq \Tr\left[M(p)^\beta\right].
\end{equation}
We work on the simplex $\Delta = \{p \ge 0, \sum_x p_x = 1\}$ and its relative interior $\Delta^\circ$. Directions $h$ in the tangent space satisfy $\sum_x h_x = 0$, and the induced perturbation is $H \triangleq \sum_x h_x A_x$.

{\noindent\textbf{Standard assumption.}}
Throughout the paper we assume
\begin{equation}
\label{eq:standing_fullsupport}
\text{(non-singularity)}: \; \sum_{x\in\cX} W_x \succ 0,
\end{equation}
equivalently $\sum_{x\in\cX} A_x \succ 0$ since $W_x$ and $W_x^\alpha$ have the same support.
This assumption is natural; otherwise one may restrict the output space to the union of the supports of $\{W_x\}_{x\in\cX}$.

\begin{assumption}[Truncated simplex (only when needed)]
\label{assump:truncate}
Fix $\delta\in(0,1/|\cX|]$ and define $\Delta_\delta\triangleq\{p\in\Delta:\,p_x\ge\delta\}$.
Whenever we invoke $\delta$-dependent constants (e.g.\ for $1<\beta<2$), we assume the iterates satisfy
$p^t\in\Delta_\delta$ for all $t$.
\end{assumption}

\begin{remark}[Regularity and the role of $\delta$]
\label{rem:regularity}
For $\beta>1$, the map $M\mapsto \Tr\left(M^\beta\right)$ is convex on the PSD cone, and standard Fr\'echet derivative formulas
apply on $\{M\succ 0\}$. When $1<\beta<2$, powers such as $M^{\beta-2}$ involve negative exponents and are singular at eigenvalue $0$.
Under~\eqref{eq:standing_fullsupport} and Assumption~\ref{assump:truncate}, we have
\[
M(p)\succeq \delta\sum_{x\in\cX}A_x \succ 0 \quad \forall\,p\in\Delta_\delta,
\]
so $\lambda_{\min}(M(p))$ is uniformly lower bounded on $\Delta_\delta$ and these expressions are well-defined.
We denote
\begin{equation}
\label{eq:def_mdelta}
m_\delta \triangleq \delta\,\lambda_{\min}\!\left(\sum_{x\in\cX}A_x\right)>0.
\end{equation}
\end{remark}

\subsection{Convex Reformulation}
Since $\frac{\alpha}{\alpha-1} < 0$ and $\log(\cdot)$ is strictly increasing, maximizing \eqref{eq:intro_Ialpha} is equivalent to minimizing $S(p)$.

\begin{lemma}[Convexity {\cite[Theorem 2.10]{Car09}}]
\label{lem:convexS}
For $\beta\ge 1$, the map $p\mapsto S(p)=\Tr\left[M(p)^\beta\right]$ is convex on $\Delta$.
\end{lemma}

\begin{remark}[Capacity via $S$]
\label{rem:cap_via_S}
The minimizers satisfy $\arg\max_{p\in\Delta} I_\alpha(p;W)=\arg\min_{p\in\Delta} S(p)$, and
\[
C_\alpha(W)=\frac{\alpha}{\alpha-1}\log\Big(\min_{p\in\Delta} S(p)\Big).
\]
\end{remark}

\section{The Algorithm}
\label{sec:algorithm}

We utilize mirror descent with the negative entropy mirror map
$\omega(p)\triangleq \sum_x p_x\log p_x$, which induces the
Kullback--Leibler divergence $\KL(p\|r)$.

\subsection{Mirror-descent update and closed form}
The gradient of $S(p)$ is
\begin{equation}
\label{eq:gradS}
[\nabla S(p)]_x
= \beta\,\Tr\!\left[M(p)^{\beta-1}A_x\right],
\qquad x\in\cX.
\end{equation}
Given a stepsize $\eta>0$, the mirror descent update is
\begin{equation}
\label{eq:MD_step}
p^{t+1}\in\arg\min_{p\in\Delta}
\Big\{\ip{\nabla S(p^t)}{p}+\tfrac{1}{\eta}\KL(p\|p^t)\Big\}.
\end{equation}

\subsection{Truncated-simplex safeguard (iterates stay in $\Delta_\delta$)}
To ensure the convergence analysis is carried out on a domain where iterates do not drift to the boundary,
we enforce a universal probability floor $\delta\in(0,1/|\cX|]$ by a simple ``smoothing'' post-processing:
after computing the exponentiated-gradient iterate $\tilde p^{t+1}\in\Delta$, set
\begin{equation}
\label{eq:smoothing_step}
p^{t+1} \triangleq (1-|\cX|\delta)\,\tilde p^{t+1} + \delta\,\mathbf{1},
\end{equation}
where $\mathbf{1}$ is the all-ones vector. Then $p^{t}\in\Delta_\delta$ for all $t$.
In other words, Assumption~\ref{assump:truncate} holds along the iterates, and under \eqref{eq:standing_fullsupport}
this implies $M(p^t)\succ 0$ for all $t$.

\subsection{Constant stepsize selection}
In this work we use a \emph{constant} stepsize $\eta=1/L$, where $L$ is any
(relative) smoothness constant of $S$ with respect to $\omega$.
Section~\ref{sec:smoothness} provides explicit bounds (Theorem~\ref{thm:rel_smooth}).
In particular:
\begin{equation}
\label{eq:L_const_qge2}
L = 2c_\beta\,\beta(\beta-1),
\qquad \beta\ge 2,
\end{equation}
where $c_\beta$ is defined in \eqref{eq:cq_def} later.
Thus for $\alpha\le \tfrac12$ the stepsize $\eta=1/L$ is explicit and global.
For $1<\beta<2$, Theorem~\ref{thm:rel_smooth} yields
$L(p)=2c_\beta\,\beta(\beta-1)\lambda_{\min}(M(p))^{\beta-2}$; hence on $\Delta_\delta$
under Assumption~\ref{assump:truncate} and \eqref{eq:standing_fullsupport} one has $\lambda_{\min}(M(p))\ge m_\delta$ as in \eqref{eq:def_mdelta},
and may take
\begin{equation}
\label{eq:L_const_qlt2}
L_\delta = 2c_\beta\,\beta(\beta-1)\,m_\delta^{\,\beta-2},
\qquad 1<\beta<2,
\end{equation}
again yielding a constant (region-dependent) stepsize $\eta=1/L_\delta$.

\begin{algorithm}[!ht]
\caption{Mirror Descent for $\alpha$-R\'enyi Capacity (with $\Delta_\delta$ safeguard)}
\label{alg:renyi_ba}
\begin{algorithmic}[1]
\State \textbf{Input:} Classical-quantum channel $W$ with $\sum_{x\in\mathcal{X}} W_x \succ 0$, $\alpha \in (0,1)$, stepsize $\eta$ (constant, e.g.,~$\eta=1/L$), floor $\delta\in(0,1/|\cX|]$.
\State \textbf{Set:} $\beta = 1/\alpha$, $A_x = W_x^\alpha$.
\State \textbf{Initialize:} $p^0 = \mathrm{uniform}(\cX)$ \hfill (so $p^0\in\Delta_\delta$ if $\delta\le 1/|\cX|$)
\State \textbf{For} $t = 0,1,2,\dots$ until convergence:
\State \quad $M_t = \sum_{x\in\cX} p_x^t A_x$.
\State \quad $v_x^t = \beta\,\Tr\left[M_t^{\beta-1} A_x\right]$ for all $x$.
\State \quad Compute the exponentiated-gradient iterate $\tilde p^{t+1}$:
\[
\tilde p_x^{t+1}
=\frac{p_x^t \exp(-\eta v_x^t)}{\sum_{y\in\cX} p_y^t \exp(-\eta v_y^t)}.
\]
\State \quad \textbf{Safeguard:} $p^{t+1} \leftarrow (1-|\cX|\delta)\,\tilde p^{t+1} + \delta\,\mathbf{1}$ \hfill (so $p^{t+1}\in\Delta_\delta$)
\State \textbf{Output:} $C_\alpha(W) \approx \frac{\alpha}{\alpha-1}\log S(p^t)$.
\end{algorithmic}
\end{algorithm}

\section{Analysis: Hessian and Kernel Bounds}
\label{sec:analysis_hessian}

We derive bounds on the Hessian to establish relative smoothness and relative strong convexity.

\subsection{Exact Directional Hessian}
When $1<\beta<2$ we will work on $\Delta_\delta$ (Assumption~\ref{assump:truncate}) so that $M(p)\succ 0$.
Let $M(p)=\sum_i \lambda_i P_i$ be the spectral decomposition.
The directional derivative is $\D S[p](h) = \beta \Tr\left[M^{\beta-1}H\right]$.

\begin{lemma}[Directional Hessian]
\label{lem:hess_exact}
For $\beta > 1$ and $M(p)\succ 0$,
\begin{equation}
\label{eq:hess_exact}
\D^2 S[p][h,h]
\!=\! \beta\sum_{i,j} g_{ij}\,\Tr\left[P_i H P_j H\right]
\!=\! \beta\sum_{i,j} g_{ij}\,\|P_i H P_j\|_F^2,
\end{equation}
where the divided-difference coefficients are
\begin{equation}
\label{eq:gij_def}
g_{ij}\triangleq
\begin{cases}
\dfrac{\lambda_i^{\beta-1}-\lambda_j^{\beta-1}}{\lambda_i-\lambda_j}, & i\neq j,\\[0.6ex]
(\beta-1)\lambda_i^{\beta-2}, & i=j.
\end{cases}
\end{equation}
In particular, $g_{ij}\ge 0$ and hence $\D^2S[p](h,h)\ge 0$.
\end{lemma}

\begin{proof}
This follows from the Daleckii--Krein formula \cite[Thm.~3.25]{HP14} for the Fr\'echet derivative 
and $\Tr\left[P_i H P_j H\right] = \|P_iHP_j\|_F^2$ with Frobenius norm $\|X\|_F\triangleq \sqrt{\Tr[X^\dagger X] }$.
\end{proof}

\subsection{Kernel Bounds}

\begin{lemma}[Integral Representation]
\label{lem:integral_rep}
For $\beta>1$ and $a,b>0$:
\begin{equation}
\label{eq:integral_rep}
\frac{a^{\beta-1}-b^{\beta-1}}{a-b}
= (\beta-1)\int_0^1\big(\theta a+(1-\theta)b\big)^{\beta-2}\,d\theta,
\end{equation}
with the natural continuous extension $(\beta-1)a^{\beta-2}$ when $a=b$.
\end{lemma}

\begin{proof}
Let $\phi(t)=t^{\beta-1}$. Then $\phi(a)-\phi(b) = \int_b^a \phi'(t)\,dt = (\beta-1)\int_b^a t^{\beta-2} dt$.
Then, substitute $t=\theta a + (1-\theta)b$.
\end{proof}

\begin{lemma}[Uniform Upper Bound]
\label{lem:kernel_upper}
Let $r \triangleq \beta-2$. For $(i,j)$ with eigenvalues $\lambda_i, \lambda_j > 0$,
\begin{equation}
\label{eq:gij_upper}
g_{ij}\le (\beta-1)c_\beta\big(\lambda_i^{r}+\lambda_j^{r}\big),
\end{equation}
where
\begin{equation}
\label{eq:cq_def}
c_\beta\triangleq
\begin{cases}
\frac12, & \beta\in(1,2]\cup[3,\infty),\\
2^{2-\beta}, & \beta\in(2,3).
\end{cases}
\end{equation}
\end{lemma}

\begin{proof}
By Lemma~\ref{lem:integral_rep},
\[
g_{ij}=(\beta-1)\int_0^1(\theta\lambda_i+(1-\theta)\lambda_j)^r\,d\theta.
\]
If $\beta\in(1,2]\cup[3,\infty)$ then $r\in[-1,0]\cup[1,\infty)$ and $t\mapsto t^r$ is convex on $(0,\infty)$.
Hermite--Hadamard's inequality implies the integral is at most $\frac{\lambda_i^r+\lambda_j^r}{2}$, giving $c_\beta=\frac12$.
If $\beta\in(2,3)$ then $r\in(0,1)$ and $t\mapsto t^r$ is concave. Jensen's inequality gives
$\int_0^1 f \le f(\frac12)$ for concave $f$, hence
\begin{align*}
\int_0^1(\theta\lambda_i+(1-\theta)\lambda_j)^r\,d\theta
&\le \Big(\tfrac{\lambda_i+\lambda_j}{2}\Big)^r
=2^{-r}(\lambda_i+\lambda_j)^r \\
&\le 2^{-r}(\lambda_i^r+\lambda_j^r),
\end{align*}
using $(a+b)^r\le a^r+b^r$ for $r\in(0,1)$. Since $2^{-r}=2^{2-\beta}$, this yields \eqref{eq:gij_upper}.
\end{proof}

\begin{lemma}[Lower Bound]
\label{lem:kernel_lower}
For $(i,j)$ with eigenvalues $\lambda_i, \lambda_j > 0$,
\begin{equation}
\label{eq:gij_lower}
g_{ij}\ge(\beta-1)\times
\begin{cases}
\lambda_{\min}(M)^{\beta-2}, & \beta\ge 2,\\
\lambda_{\max}(M)^{\beta-2}, & 1<\beta\le 2.
\end{cases}
\end{equation}
\end{lemma}

\begin{proof}
By Lemma~\ref{lem:integral_rep}, the integral is lower bounded by the minimum of its integrand on $[0,1]$.
If $\beta\ge 2$, then $t^{\beta-2}$ is increasing, so the minimum over the segment between $\lambda_i$ and $\lambda_j$
is attained at $\min\{\lambda_i,\lambda_j\}\ge\lambda_{\min}(M)$.
If $1<\beta\le 2$, then $t^{\beta-2}$ is decreasing, so the minimum is attained at $\max\{\lambda_i,\lambda_j\}\le\lambda_{\max}(M)$.
\end{proof}

\section{Relative Smoothness and Global Convergence}
\label{sec:smoothness}

\begin{definition}[Relative smoothness / strong convexity]
\label{def:relcurv}
Let $\omega$ be differentiable and $D_\omega(p',p)\triangleq \omega(p')-\omega(p)-\ip{\nabla\omega(p)}{p'-p}$.
A differentiable function $f$ is \emph{$L$-smooth relative to $\omega$} if
\[
f(p')\le f(p)+\ip{\nabla f(p)}{p'-p}+L\,D_\omega(p',p)\quad \forall p,p'.
\]
It is \emph{$\mu$-strongly convex relative to $\omega$} if
\[
f(p')\ge f(p)+\ip{\nabla f(p)}{p'-p}+\mu\,D_\omega(p',p)\quad \forall p,p'.
\]
For $\omega(p)=\sum_x p_x\log p_x$, one has $D_\omega(p',p)\!=\!\KL(p'\|p)$ and
$\D^2\omega[p](h,h)=\sum_x h_x^2/p_x$ on the tangent space $\sum_x h_x=0$.
\end{definition}

\begin{remark}[Hessian dominance implies relative curvature]
\label{rem:hess_to_rel}
If $f$ and $\omega$ are twice continuously differentiable on $\Delta^\circ$ and for all $p\in\Delta^\circ$ and tangent $h$
one has $\D^2 f[p](h,h)\le L\,\D^2\omega[p](h,h)$, then $f$ is $L$-smooth relative to $\omega$ in the sense of
Definition~\ref{def:relcurv} (and similarly, $\D^2 f[p](h,h)\ge \mu\,\D^2\omega[p](h,h)$ implies $\mu$-relative strong convexity),
by integrating the second-order bound along line segments in $\Delta^\circ$.
\end{remark}

\begin{theorem}[Relative Smoothness]
\label{thm:rel_smooth}
Let $\beta>1$ and $S(p)=\Tr\left[M(p)^\beta\right]$.
If $\beta\ge 2$, then for all $p\in\Delta$ and all tangent vectors $h$ (i.e.\ $\sum_x h_x=0$),
\begin{equation}
\label{eq:rel_smooth_hess}
\D^2S[p](h,h)\le L(p)\sum\nolimits_x \frac{h_x^2}{p_x},
\end{equation}
with
\begin{equation}
\label{eq:L_of_p}
L(p)\triangleq
2c_\beta\,\beta(\beta-1) \times
\begin{cases}
1, & \beta\ge 2,\\[0.3ex]
\lambda_{\min}(M(p))^{\beta-2}, & 1<\beta<2.
\end{cases}
\end{equation}
Consequently, $S$ is $L$-smooth relative to the negative entropy $\omega$ in the sense of Definition~\ref{def:relcurv}
(with constant $L=2c_\beta\,\beta(\beta-1)$ for $\beta\ge 2$).
If $1<\beta<2$, the same bound holds for all $p$ with $M(p)\succ 0$; moreover on $\Delta_\delta$
under Assumption~\ref{assump:truncate} and \eqref{eq:standing_fullsupport} one may take the constant $L_\delta$
in \eqref{eq:L_const_qlt2}.
\end{theorem}

\begin{proof}
From Lemma~\ref{lem:hess_exact} and Lemma~\ref{lem:kernel_upper},
\[
\D^2S[p][h,h]
\le \beta(\beta-1)c_\beta\sum_{i,j}(\lambda_i^{r}+\lambda_j^{r})\,\Tr\left[P_iHP_jH\right].
\]
Since $M(p)\succ 0$, the spectral projectors satisfy $\sum_j P_j=I$. Using cyclicity of trace,
\begin{align*}
\sum_{i,j}\lambda_i^r \Tr\left[P_iHP_jH\right]
&=\sum_i \lambda_i^r \Tr\!\left[P_i H\Big(\sum_j P_j\Big)H\right] \\
&=\sum_i \lambda_i^r \Tr\left[P_i H^2\right] \\
&=\Tr\left[M^r H^2\right]=\Tr\left[HM^rH\right],
\end{align*}
and similarly $\sum_{i,j}\lambda_j^r \Tr\left[P_iHP_jH\right]=\Tr\left[HM^rH\right]$, hence
\begin{equation}
\label{eq:HM_power_bound}
\D^2S[p](h,h)\le 2c_\beta\,\beta(\beta-1)\,\Tr\!\left[HM^{\beta-2}H\right].
\end{equation}
Let $v_x = M^{\frac{\beta-2}{2}}A_x$. Then $\Tr\left[HM^{\beta-2}H\right]=\|\sum_x h_x v_x\|_F^2$.
By Cauchy--Schwarz with weights $p_x$,
\begin{equation}
\label{eq:cs_weighted}
\Big\|\sum_x h_x v_x\Big\|_F^2
\le \Big(\sum_x \frac{h_x^2}{p_x}\Big)\underbrace{\Big(\sum_x p_x \|v_x\|_F^2\Big)}_{B(p)}.
\end{equation}
We bound $B(p)=\sum_x p_x\Tr\left[A_x M^{\beta-2}A_x\right]$.
Assume each $W_x$ is a density operator, so $\Tr\left[W_x\right]=1$. Then
\[
\|A_x\|_\beta^\beta=\Tr\left[A_x^\beta\right]=\Tr\left[W_x\right]=1
\quad\Rightarrow\quad
\|A_x\|_\beta=1.
\]
\emph{Case $\beta\ge 2$.} Applying H\"older's inequality with exponents $(\beta,\frac{\beta}{\beta-2},\beta)$ yields
\begin{align*}
\Tr\left[A_x M^{\beta-2} A_x\right]
&\le \|A_x\|_\beta^2\,\|M^{\beta-2}\|_{\frac{\beta}{\beta-2}} \\
&= \|M\|_\beta^{\beta-2}=S(p)^{\frac{\beta-2}{\beta}}.
\end{align*}
Moreover $\|M\|_\beta\le \sum_x p_x\|A_x\|_\beta=1$, hence $S(p)\le 1$, so $B(p)\le 1$.
\emph{Case $1<\beta<2$.} Since $\beta-2<0$, $t\mapsto t^{\beta-2}$ is decreasing and
$\|M^{\beta-2}\|_\infty=\lambda_{\min}(M)^{\beta-2}$. Therefore
\begin{align*}
\Tr\left[A_x M^{\beta-2}A_x\right]
&\le \lambda_{\min}(M)^{\beta-2}\Tr\left[A_x^2\right] \\
&=\lambda_{\min}(M)^{\beta-2}\Tr\left[W_x^{2\alpha}\right].
\end{align*}
As $\beta<2 \iff \alpha>1/2$, we have $2\alpha>1$, and for eigenvalues $\lambda\in[0,1]$ one has $\lambda^{2\alpha}\le \lambda$,
so $\Tr\left[W_x^{2\alpha}\right]\le \Tr\left[W_x\right]=1$. Hence $B(p)\le \lambda_{\min}(M)^{\beta-2}$.
Combining \eqref{eq:HM_power_bound}--\eqref{eq:cs_weighted} yields \eqref{eq:rel_smooth_hess} with \eqref{eq:L_of_p}.
\end{proof}

\begin{theorem}[Global Sublinear Rate (best iterate)]
\label{thm:sublinear}
Assume $S$ is convex and $L$-smooth relative to $\omega$ with a \emph{constant} $L$ along the iterates
(e.g.\ $\beta\ge 2$, or $1<\beta<2$ on $\Delta_\delta$ with $L=L_\delta$),
and run Algorithm~\ref{alg:renyi_ba} with $\eta=1/L$.
Then
\[
\min_{0\le t\le T-1}\big(S(p^{t+1})-S(p^\star)\big)
\le \frac{L\,\KL(p^\star\|p^0)}{T},
\]
where $p^\star$ is a minimizer of $S$ over $\Delta$.
\end{theorem}

\begin{corollary}[Iteration complexity (sublinear)]
If $\beta\ge 2$ (so $L=2c_\beta \beta(\beta-1)$ is constant), then to achieve
$\min_{t<T} (S(p^{t})-S(p^\star))\le \varepsilon$ it suffices that
\[
T = O\!\left(\frac{L\,\KL(p^\star\|p^0)}{\varepsilon}\right).
\]
\end{corollary}

\section{Relative Strong Convexity and Linear Rate}
\label{sec:linear}

We now prove local linear convergence on $\Delta_\delta \triangleq \{p \in \Delta : p_x \ge \delta\}$.
Let $T=\{h\in\mathbb{R}^{|\cX|}:\sum_x h_x=0\}$ denote the simplex tangent space.

\begin{lemma}[Tangent-Space Nondegeneracy]
\label{lem:tangent_nondeg}
Let $G_{xy} = \Tr\left[A_x A_y\right]$ be the Frobenius Gram matrix. Define
\[
\gamma \triangleq \lambda_{\min}(G|_T)
= \inf_{\substack{h \in T\\h \ne 0}} \frac{h^\top G h}{\|h\|_2^2}.
\]
Then $\Tr\left[H^2\right]=\|\sum_x h_xA_x\|_F^2 = h^\top G h \ge \gamma \|h\|_2^2$.
Moreover, $\gamma>0$ iff the map $h\mapsto \sum_x h_xA_x$ is injective on $T$.
\end{lemma}

\begin{proof}
$G$ is a Gram matrix, hence positive semi-definite and $h^\top G h=\|\sum_x h_xA_x\|_F^2$.
Restricting to $T$ and applying the Rayleigh quotient yields the bound and characterization.
\end{proof}

\begin{theorem}[Relative Strong Convexity on $\Delta_\delta$]
\label{thm:rel_strong}
Assume $\gamma>0$ from Lemma~\ref{lem:tangent_nondeg} and fix $\delta\in(0,1/|\cX|]$.
Let $p\in\Delta_\delta$ and $h\in T$.
\begin{enumerate}
\item If $1<\beta\le 2$, then
\begin{equation}
\label{eq:rsc_1}
\D^2S[p](h,h) \ge \beta(\beta-1)\,\gamma\,\delta \sum\nolimits_x \frac{h_x^2}{p_x}.
\end{equation}
\item If $\beta>2$, assume additionally that $\sum_{x\in\cX} A_x\succ 0$. Then for all $p\in\Delta_\delta$,
\[
M(p) \succeq \delta\sum_x A_x
\Rightarrow
\lambda_{\min}(M(p))\ge m_\delta >0,
\]
where $m_\delta \triangleq \delta\,\lambda_{\min}\Big(\sum_x A_x\Big)$, and
\begin{equation}
\label{eq:rsc_2}
\D^2S[p](h,h) \ge \beta(\beta-1)\,m_\delta^{\beta-2}\,\gamma\,\delta \sum\nolimits_x \frac{h_x^2}{p_x}.
\end{equation}
\end{enumerate}
Hence, on $\Delta_\delta$, $S$ is $\mu$-strongly convex relative to $\omega$ with $\mu$ given by the
right-hand side coefficient in \eqref{eq:rsc_1} or \eqref{eq:rsc_2}.
\end{theorem}

\begin{proof}
For $p\in\Delta_\delta$ under Assumption~\ref{assump:truncate} and \eqref{eq:standing_fullsupport}, we have $M(p)\succ 0$.
From Lemma~\ref{lem:hess_exact} and Lemma~\ref{lem:kernel_lower},
\begin{align*}
\D^2S[p](h,h)
&= \beta\sum_{i,j} g_{ij}\,\Tr\left[P_iHP_jH\right] \\
&\ge \beta(\beta-1)\kappa(p)\sum_{i,j}\Tr\left[P_iHP_jH\right] \\
&= \beta(\beta-1)\kappa(p)\Tr\left[H^2\right],
\end{align*}
where $\kappa(p)=\lambda_{\max}(M(p))^{\beta-2}$ for $1<\beta\le 2$ and $\kappa(p)=\lambda_{\min}(M(p))^{\beta-2}$ for $\beta\ge 2$.
By Lemma~\ref{lem:tangent_nondeg}, $\Tr\left[H^2\right]\ge \gamma\|h\|_2^2$.
For $p\in\Delta_\delta$,
\[
\|h\|_2^2=\sum_x h_x^2 \ge \delta\sum\nolimits_x \frac{h_x^2}{p_x}.
\]
If $1<\beta\le 2$, then $\|A_x\|_\infty\le 1$ implies $\lambda_{\max}(M(p))\le \sum_x p_x\|A_x\|_\infty\le 1$,
and since $\beta-2\le 0$ we get $\kappa(p)=\lambda_{\max}(M(p))^{\beta-2}\ge 1$, yielding \eqref{eq:rsc_1}.
If $\beta>2$ and $\sum_x A_x\succ 0$, then $\lambda_{\min}(M(p))\ge m_\delta$, so $\kappa(p)\ge m_\delta^{\beta-2}$,
yielding \eqref{eq:rsc_2}.
\end{proof}

\begin{theorem}[Linear Convergence in KL]
\label{thm:linear}
Assume that on $\Delta_\delta$ the function $S$ is $L$-smooth and $\mu$-strongly convex relative to $\omega$,
and that the iterates produced by Algorithm~\ref{alg:renyi_ba} remain in $\Delta_\delta$ (guaranteed by \eqref{eq:smoothing_step}).
Let $p^\star$ denote the (unique) minimizer of $S$ over $\Delta_\delta$ and set $\eta=1/L$.
Then
\begin{equation}
\label{eq:linear_KL}
\KL(p^\star\|p^{t+1})
\le \Big(1-\frac{\mu}{L}\Big)\KL(p^\star\|p^{t}).
\end{equation}
\end{theorem}

\begin{proof}
\textbf{(i) Three-point inequality.}
The optimality condition for \eqref{eq:MD_step} implies for any $u\in\Delta$,
\begin{multline*}
\ip{\nabla S(p^t)}{p^{t+1}-u} \\
\le \frac{1}{\eta}\Big(\KL(u\|p^t) - \KL(u\|p^{t+1}) - \KL(p^{t+1}\|p^t)\Big).
\end{multline*}
\textbf{(ii) Relative smoothness descent.}
Relative smoothness gives
\[
S(p^{t+1}) \le S(p^t) + \ip{\nabla S(p^t)}{p^{t+1}-p^t} + \frac{1}{\eta}\KL(p^{t+1}\|p^t).
\]
Using the three-point inequality with $u=p^\star$ and decomposing
$\ip{\nabla S(p^t)}{p^{t+1}-p^t}=\ip{\nabla S(p^t)}{p^{t+1}-p^\star}+\ip{\nabla S(p^t)}{p^\star-p^t}$,
the $\KL(p^{t+1}\|p^t)$ terms cancel and we obtain
\begin{align*}
S(p^{t+1})
&\le S(p^t)+\ip{\nabla S(p^t)}{p^\star-p^t} \\
&\quad +\frac{1}{\eta}\Big(\KL(p^\star\|p^t)-\KL(p^\star\|p^{t+1})\Big).
\end{align*}
\textbf{(iii) Relative strong convexity.}
Relative strong convexity gives
\[
S(p^\star)\ge S(p^t)+\ip{\nabla S(p^t)}{p^\star-p^t}+\mu\,\KL(p^\star\|p^t),
\]
or equivalently $S(p^t)+\ip{\nabla S(p^t)}{p^\star-p^t}\le S(p^\star)-\mu\KL(p^\star\|p^t)$.
Substituting yields
\begin{align*}
S(p^{t+1})
&\le S(p^\star) - \mu\KL(p^\star\|p^t) \\
&\quad +L\Big(\KL(p^\star\|p^t)-\KL(p^\star\|p^{t+1})\Big).
\end{align*}
Since $p^\star$ minimizes $S$ over $\Delta_\delta$ and $p^{t+1}\in\Delta_\delta$, we have $S(p^{t+1})\ge S(p^\star)$,
so dropping $S(p^{t+1})-S(p^\star)\ge 0$ gives \eqref{eq:linear_KL}.
\end{proof}

\begin{corollary}[Iteration complexity (linear)]
Under the assumptions of Theorem~\ref{thm:linear}, to achieve $\KL(p^\star\|p^{t})\le \varepsilon$ it suffices that
\[
t = O\!\left(\frac{L}{\mu}\log\frac{\KL(p^\star\|p^0)}{\varepsilon}\right).
\]
\end{corollary}

\section{Numerical Experiments}
\label{sec:num}

We evaluate the proposed mirror-descent Blahut--Arimoto iteration (Algorithm~\ref{alg:renyi_ba}) on a randomly generated \emph{noncommuting} c-q channel. The goal is to illustrate (i) the computed Petz--R\'enyi capacity across $\alpha$, (ii) the convergence behavior of our iterates, and (iii) basic scaling with the problem size.

\subsection{Experimental setup}
We consider a c-q channel $W:\cX\to\Dens(\mathbb{C}^d)$ with $|\cX|=10$ and $d=6$.
Each output state is generated randomly (noncommuting) and then mixed with the maximally mixed state:
\[
W_x \leftarrow (1-\varepsilon)\rho_x + \varepsilon \frac{I}{d},
\qquad \varepsilon = 10^{-2},
\]
ensuring full rank and numerical stability. We sweep $\alpha\in\{0.10,0.20,\dots,0.90\}$ and run Algorithm~\ref{alg:renyi_ba} with
maximum $T=30000$ iterations, stopping tolerance $10^{-8}$, and probability floor $\delta=10^{-11}$ (with an adaptive stepsize option enabled in our implementation).

\subsection{Capacity sweep and runtime profile}
Fig.~\ref{fig:alpha_metrics_ours} reports the estimated capacities $C_\alpha(W)$, wall-clock runtime, and iterations-to-stop across $\alpha$.
On this noncommuting instance, the estimated capacity increases smoothly with $\alpha$ (from $\approx 0.064$ at $\alpha=0.1$ to $\approx 0.413$ at $\alpha=0.9$).
The computational effort is $\alpha$-dependent: smaller $\alpha$ values can be noticeably harder (e.g., $\alpha=0.1$ reaches the iteration cap in our run),
while mid-range $\alpha$ converges quickly (e.g., at $\alpha=0.5$ the method stops after about $1.9\times 10^3$ iterations and $\approx 0.67$s).

\begin{figure*}[!t]
    \centering
    \includegraphics[width=0.32\textwidth]{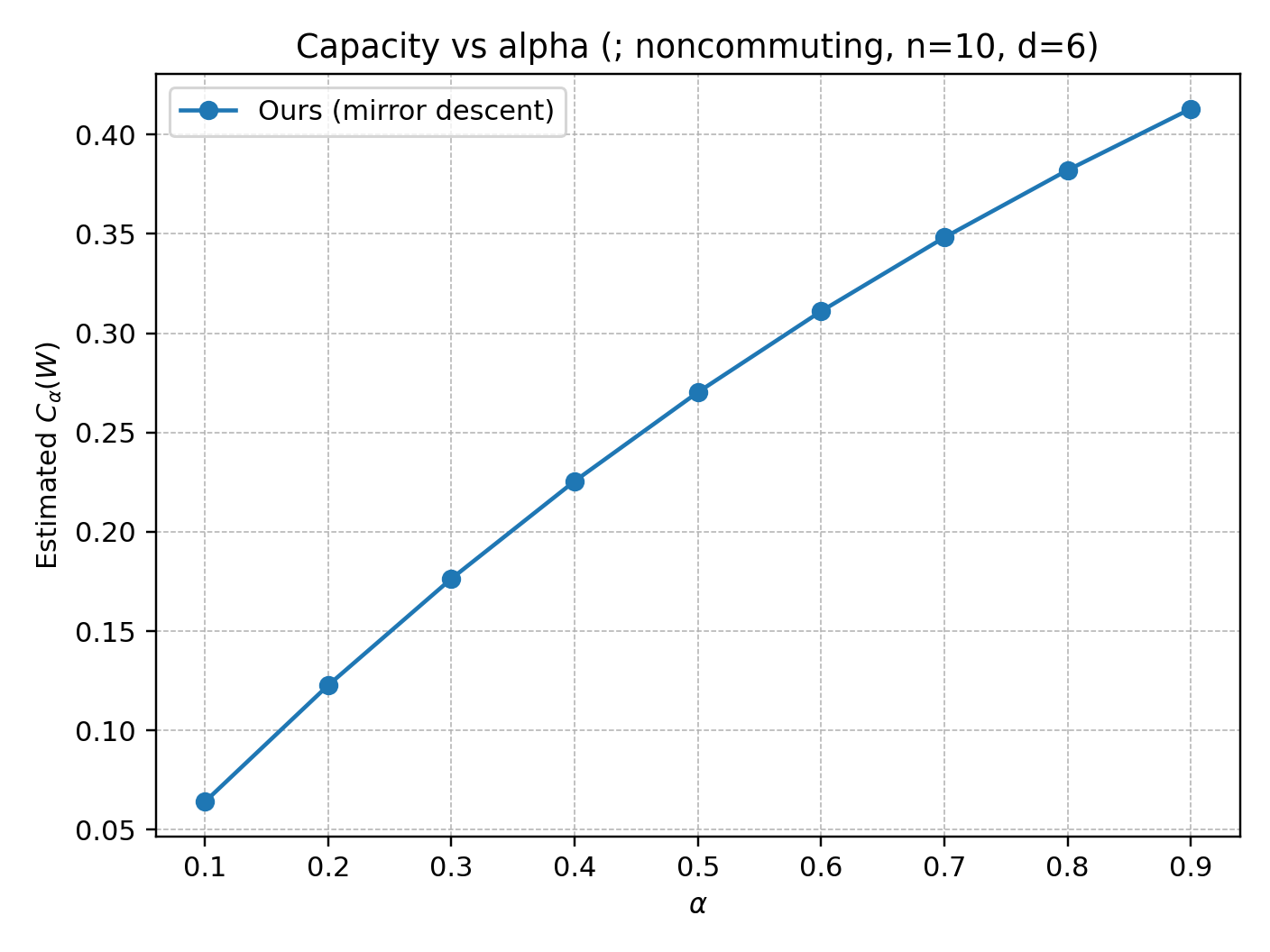}\hfill
    \includegraphics[width=0.32\textwidth]{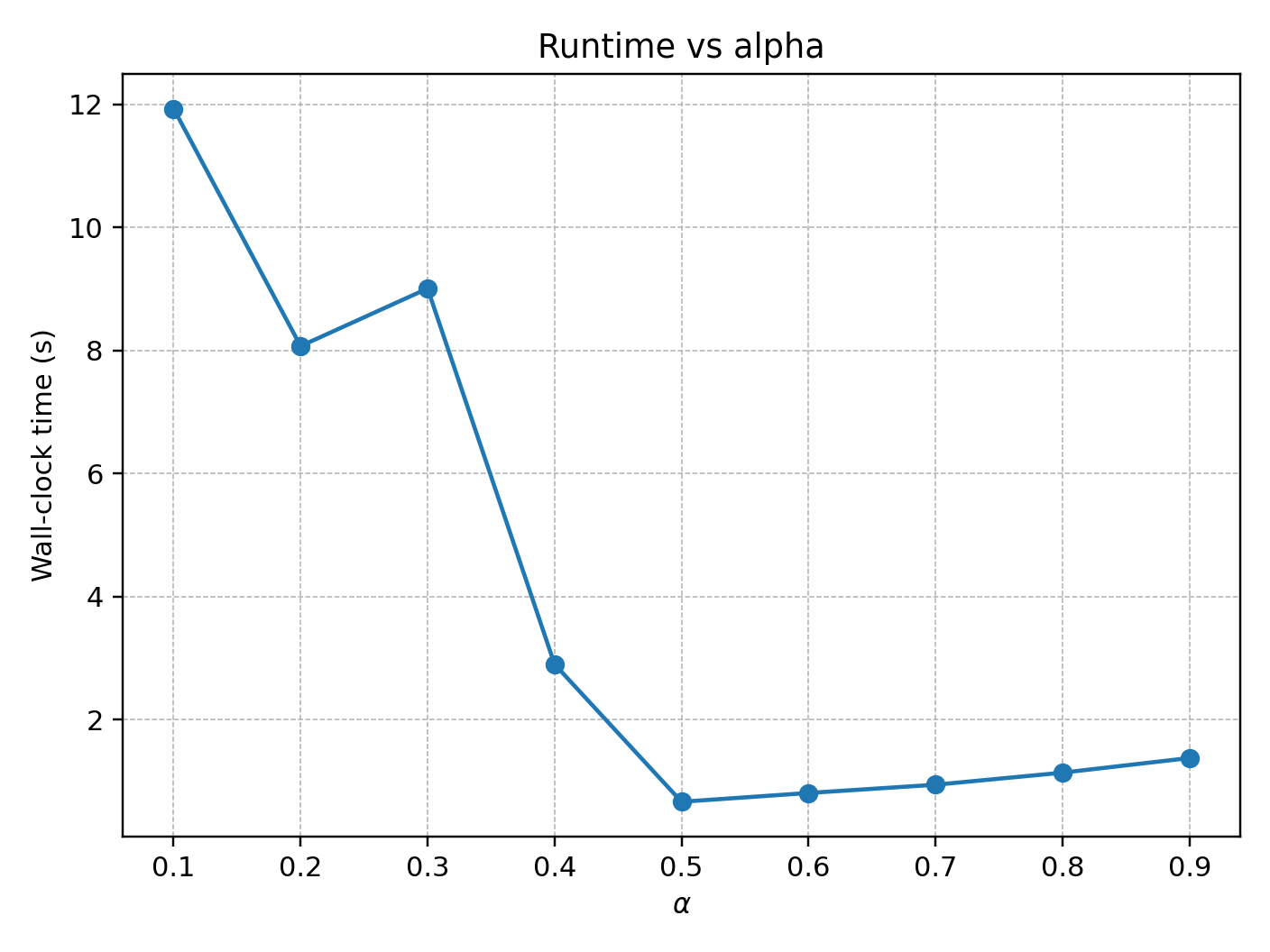}\hfill
    \includegraphics[width=0.32\textwidth]{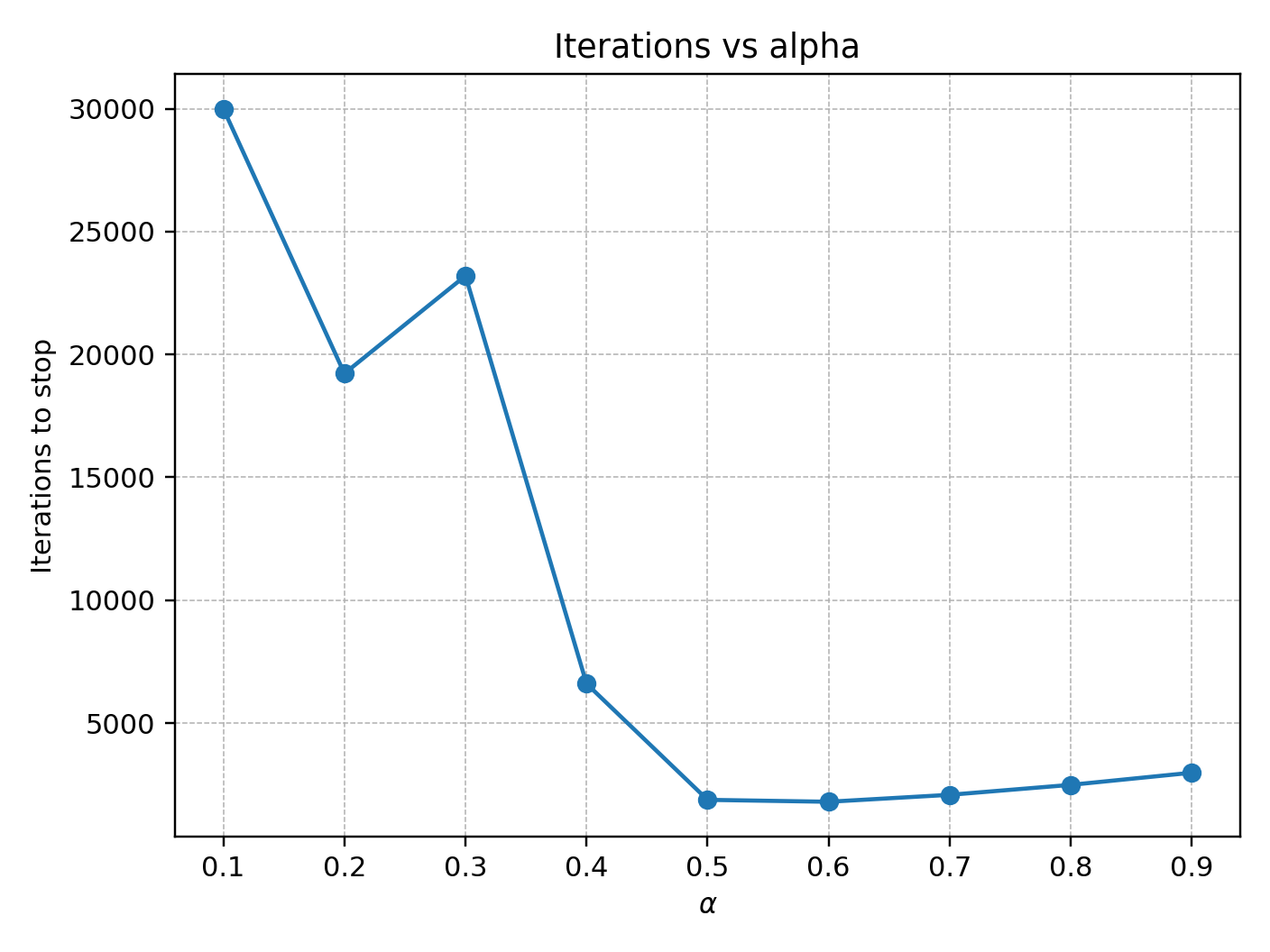}
    \caption{Evaluation on a noncommuting channel ($|\cX|=10$, $d=6$).
    (Left) Estimated Petz--R\'enyi capacity $C_\alpha(W)$ versus $\alpha$.
    (Center) Wall-clock runtime versus $\alpha$.
    (Right) Iterations-to-stop versus $\alpha$.}
    \label{fig:alpha_metrics_ours}
\end{figure*}

\subsection{Convergence certificate (duality gap)}
To monitor convergence without relying on an external baseline, we plot a standard stationarity certificate for convex optimization on the simplex: the duality
\[
g(p)\triangleq \max_{q\in\Delta}\langle \nabla S(p),\,p-q\rangle
= \langle \nabla S(p),p\rangle - \min_{x\in\cX} [\nabla S(p)]_x,
\]
which equals zero if and only if $p$ is first-order optimal for the simplex-constrained problem.
Fig.~\ref{fig:convergence_fwgap} shows $g(p^t)$ for $\alpha\in\{0.2,0.5,0.8\}$.
In all three cases the gap decreases steadily over the run, and the decay is substantially faster for $\alpha=0.5$ and $\alpha=0.8$ than for $\alpha=0.2$,
consistent with the iteration-count trends in Fig.~\ref{fig:alpha_metrics_ours}.

\begin{figure*}[!t]
    \centering
    \includegraphics[width=0.32\textwidth]{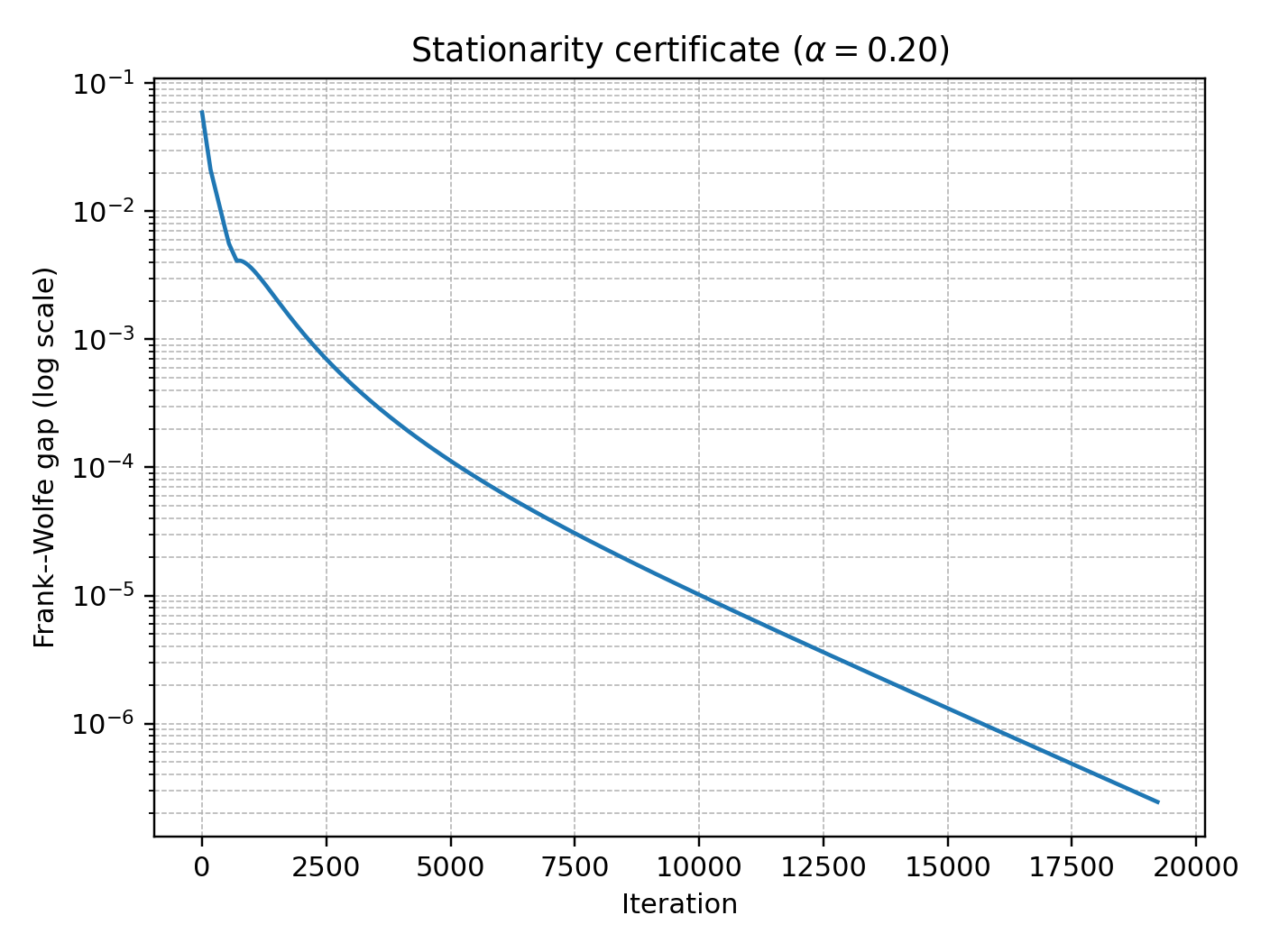}\hfill
    \includegraphics[width=0.32\textwidth]{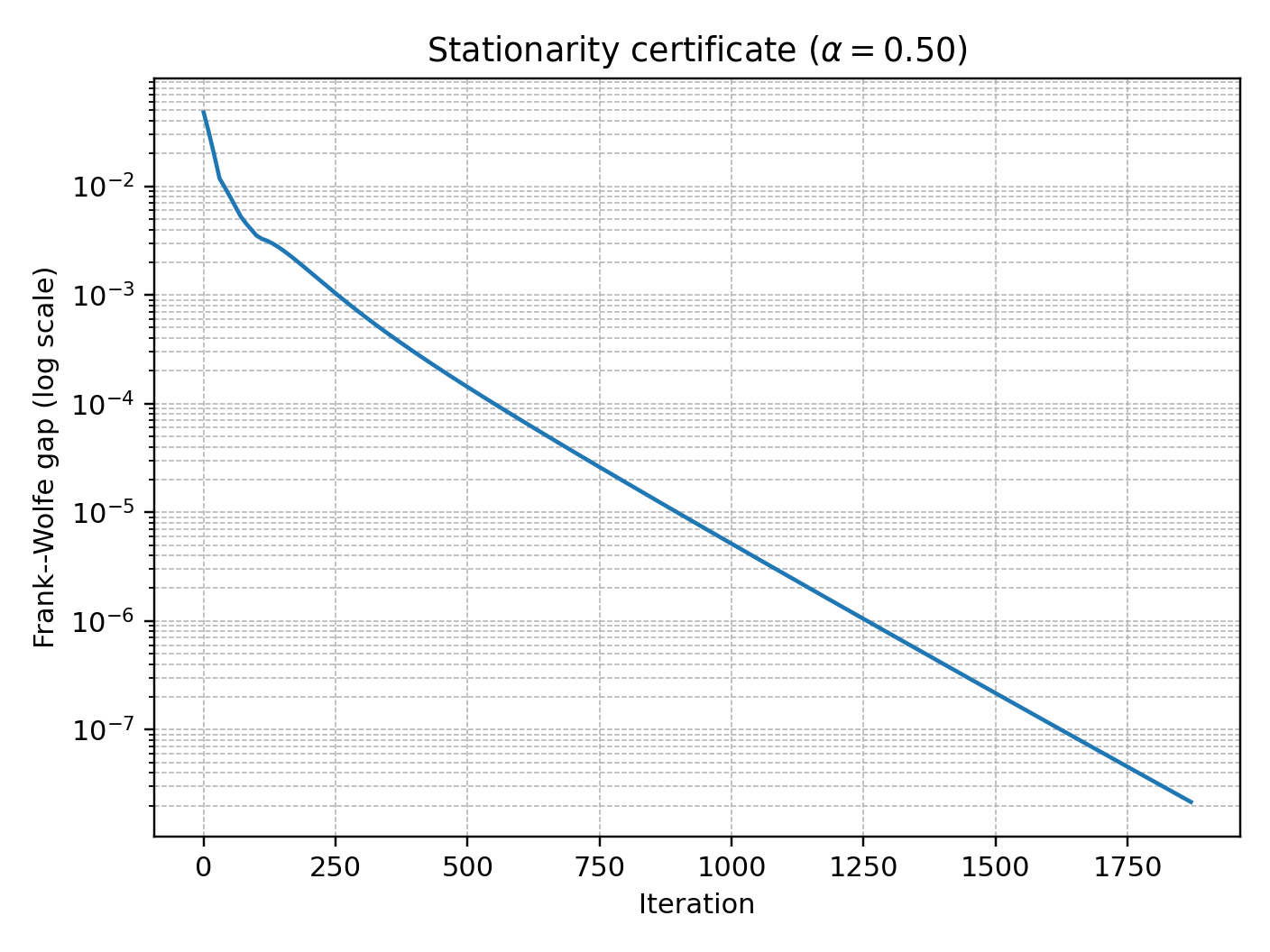}\hfill
    \includegraphics[width=0.32\textwidth]{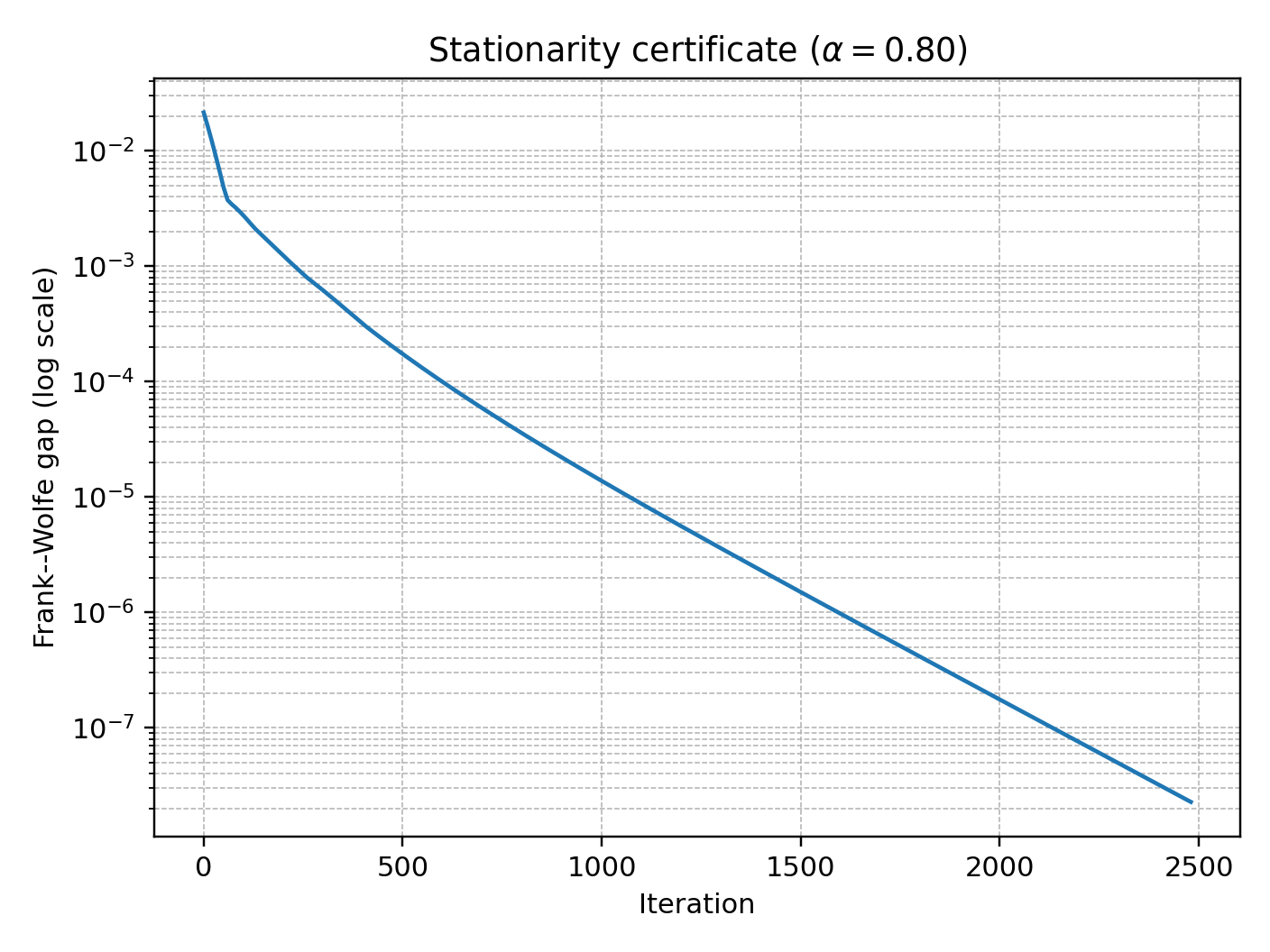}
    \caption{Convergence traces (gap) for $\alpha\in\{0.2,0.5,0.8\}$.
    The gap $g(p^t)$ provides a first-order stationarity certificate on the simplex.}
    \label{fig:convergence_fwgap}
\end{figure*}

\subsection{Scaling with $|\cX|$ and $d$}
Finally, we probe empirical scaling at a representative value $\alpha=0.5$ by varying the alphabet size $|\cX|$ (with $d=6$ fixed) and the output dimension $d$ (with $|\cX|=10$ fixed), using the same stopping tolerance and a maximum of $T=30000$ iterations (two repeats; we report median runtime).
Table~\ref{tab:scaling} summarizes the median runtimes in our sweep. As $|\cX|$ increases, the runtime rises substantially on this instance, while the dependence on $d$ remains comparatively modest for $d\le 10$.

\begin{table}[!t]
\centering
\caption{Empirical scaling at $\alpha=0.5$ (median runtime in seconds).}
\label{tab:scaling}
\scriptsize
\setlength{\tabcolsep}{4pt}
\begin{tabular}{c|cccc}
\hline
$|\cX|$ (fixed $d=6$) & 10 & 20 & 40 & 80 \\
runtime (s)           & 0.99 & 1.77 & 10.74 & 23.25 \\
\hline
$d$ (fixed $|\cX|=10$) & 4 & 6 & 8 & 10 \\
runtime (s)            & 0.96 & 0.66 & 1.06 & 0.91 \\
\hline
\end{tabular}
\end{table}

\section{Conclusion}
We cast the Petz--R\'enyi capacity computation for finite-alphabet c-q channels as minimizing a convex trace-power objective over the simplex and derive a Blahut--Arimoto--type entropic mirror-descent iteration with a closed-form exponentiated-gradient update and constant stepsize. Using divided-difference Hessian bounds, we prove entropy-relative smoothness and obtain a global sublinear rate for $\alpha\in(0,1)$, and further establish local linear convergence in KL divergence on a truncated simplex under a natural tangent-space nondegeneracy condition. Experiments on noncommuting random channels show the capacity trend across $\alpha$, verify convergence via a duality gap certificate, and demonstrate basic scaling with $|\cX|$ and $d$.

\bibliographystyle{IEEEtran}
\bibliography{reference, reference2}

\end{document}